\documentclass[aps,pra,superscriptaddress,twocolumn,floatfix,a4paper]{revtex4}

\usepackage{graphicx,graphics,epsfig,epstopdf}   
\usepackage{dcolumn}    
\usepackage{bm}         
\usepackage{amsmath}    
\usepackage{verbatim}   
\usepackage{color}      
\usepackage{subfigure}  
\usepackage{times,natbib}
\usepackage{amsmath,amsfonts,amssymb,graphics,graphics,color,times}
\usepackage{hyperref, changes}
\usepackage{bbm}

\usepackage{latexsym}
\usepackage{amsmath}
\usepackage{amssymb}
\usepackage{amsfonts}
\usepackage{amsthm}
\usepackage{mathrsfs}
\usepackage{color,verbatim,graphics}
\usepackage{psfrag}

\usepackage{epstopdf}
\DeclareMathAlphabet{\mathrsfs}{U}{rsfs}{m}{n}
\DeclareMathAlphabet{\mathpzc}{OT1}{pzc}{m}{it}
\DeclareMathAlphabet{\matheus}{U}{eus}{m}{n}
\DeclareMathAlphabet{\mathbbold}{U}{bbold}{m}{n}

\newtheorem{observation}{Observation}

\newcommand{\ba}{\begin{eqnarray}}
\newcommand{\ea}{\end{eqnarray}}
\newcommand{\ban}{\begin{eqnarray*}}
\newcommand{\ean}{\end{eqnarray*}}
\newcommand{\Tr}{\operatorname{Tr}}

\newcommand{\ket}[1]{|#1\rangle}

\newcommand{\one}{\mathbbold{1}}

\definecolor{ngreen}{rgb}{0.2,0.6,0.2}

\newcommand{\SCHSH}{\mathcal{S}_{\mbox{\tiny CHSH}}}
\newcommand{\Stau}{\mathcal{S}_\tau}

\begin{document}

\title{Exploring the limits of quantum nonlocality with entangled photons}

\author{Bradley G. Christensen}
\affiliation{Department of Physics, University of Illinois at Urbana-Champaign, Urbana, IL 61801, USA}

\author{Yeong-Cherng Liang}
\affiliation{Department of Physics, National Cheng Kung University, Tainan 701, Taiwan}
\affiliation{Institute for Theoretical Physics, ETH Z\"urich, 8093 Zurich, Switzerland}

\author{Nicolas Brunner}
\affiliation{D\'epartement de Physique Th\'eorique, Universit\'e de Gen\`eve, 1211 Gen\`eve, Switzerland}

\author{Nicolas Gisin}
\affiliation{Group of Applied Physics, University of Geneva, CH-1211 Geneva 4, Switzerland}

\author{Paul G. Kwiat}
\affiliation{Department of Physics, University of Illinois at Urbana-Champaign, Urbana, IL 61801, USA}

\date{\today}

\begin{abstract}
Quantum nonlocality is arguably among the most counter-intuitive phenomena predicted by quantum theory. In recent years, the development of an abstract theory of nonlocality has brought a much deeper understanding of the subject. In parallel, experimental progress allowed for the demonstration of quantum nonlocality in a wide range of physical systems, and brings us close to a final loophole-free Bell test. Here we combine these theoretical and experimental developments in order to explore the limits of quantum nonlocality. This approach represents a thorough test of quantum theory, and could provide evidence of new physics beyond the quantum model. Using a versatile and high-fidelity source of pairs of polarization entangled photons, we explore the boundary of quantum correlations, present the most nonlocal correlations ever reported, demonstrate the phenomenon of more nonlocality with less entanglement, and show that non-planar (and hence complex) qubit measurements can be necessary to reproduce the strong qubit correlations that we observed. Our results are in remarkable agreement with quantum predictions.
\end{abstract}

\maketitle

\section{Introduction}

Distant observers sharing a well-prepared entangled state can establish correlations which cannot be explained by any theory compatible with a natural notion of locality, as witnessed via a suitable Bell inequality violation \cite{bell}. Once viewed as marginal, nonlocality is today considered as one of the most fundamental aspects of quantum theory \cite{review,review2}, and represents a powerful resource in quantum information science, in particular in the context of the device-independent approach \cite{acin07,pironio,colbeck}. Experimental evidence is overwhelming, all major loopholes have been individually addressed, and the ultimate loophole-free Bell test is within reach \cite{christensen,giustina,hofmann,diamond}. While most Bell experiments performed so far \cite{aspect,tittel,weihs,rowe,ansmann} make use of the 
Clauser-Horne-Shimony-Holt (CHSH) Bell inequality \cite{chsh}, few exploratory works considered Bell tests in the multipartite setting \cite{pan,lavoie,erven}, or for high-dimensional systems \cite{thew,zeilinger,padgett}.

Quantum nonlocality is, however, a much richer phenomenon, explored in recent years through the development of a generalized theory of nonlocality \cite{PR,barrett,review}. The theory aims at characterizing correlations satisfying the no-signaling principle (hence not in direct conflict with relativity). Remarkably, there exist no-signaling correlations which are stronger than any correlations realizable in quantum theory \cite{PR}. The most famous example here is the maximally nonlocal Popescu-Rohrlich (PR) box. Recent works showed that such super-quantum correlations may have implausible consequences from an information-theoretic point of view \cite{vanDam,IC,LO}. Also, a more physics-based approach relying on the concept of macroscopic locality has been developed \cite{ML,rohrlich,gisin}. From these various approaches, it appears unlikely that correlations exist in nature that could violate the CHSH-Bell inequality more than what is allowed in quantum theory, that is, one recovers the Tsirelson bound \cite{tsirelson}. However, these results do not exclude stronger-than-quantum correlations for other Bell inequalities \cite{Allcock2009}, some of which we shall investigate here. This led to the intriguing concept of {\it almost quantum} correlations \cite{almostQ}. Importantly, these developments provide a fresh perspective on the foundations of quantum theory (see e.g., \cite{popescu14} for a recent review). 

\begin{figure}[b]
\includegraphics[width=\columnwidth]{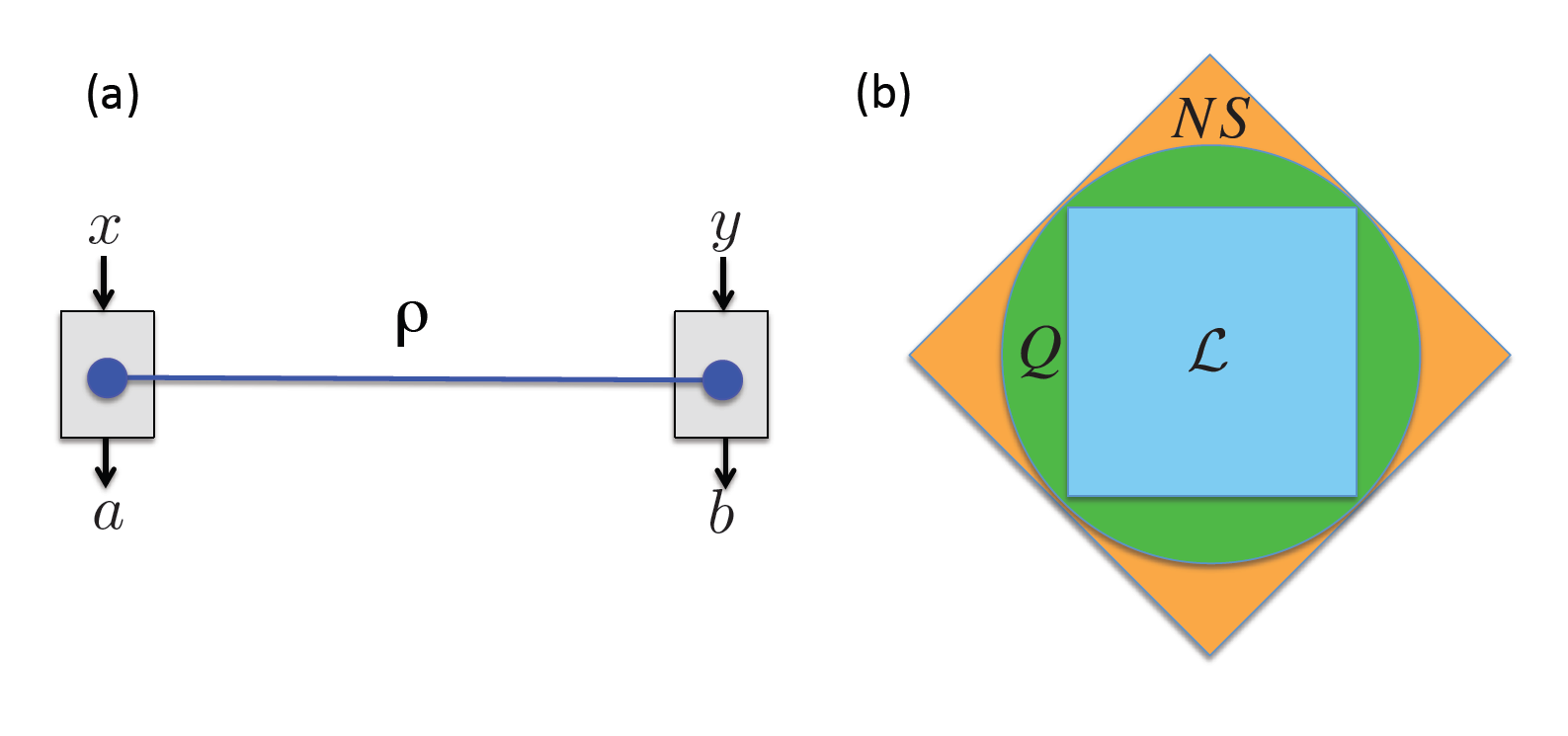}
 \caption{(a) Bell test scenario. Alice and Bob perform ``black box" measurements on a shared (quantum) state $\rho$. The experiment is characterized by the data $\{p(a,b|x,y)\}$, i.e., a set of conditional probabilities for each pair of measurement outputs ($a$ and $b$) given measurement settings $x$ and $y$. Based on the data $p(a,b|x,y)$, Bell inequalities (see Eq.~\eqref{BI}) can be tested. (b) Geometrical representation of non-signaling correlations. The set of local ($\mathcal{L}$), quantum ($\mathcal{Q}$), and non-signaling ($\mathcal{NS}$) distributions are projected onto a plane, where the following inclusion relations are clear: $\mathcal{L} \subset \mathcal{Q} \subset \mathcal{NS}$. }\label{bellpolytope}
\end{figure}

All of these recent developments have come from a theoretical perspective; here we begin to experimentally explore the limits of quantum nonlocality using a high-quality entangled-photon source. We perform a wide range of Bell tests. In particular we probe the boundary of quantum correlations in the CHSH Bell scenario. We demonstrate the phenomenon of more nonlocality with less entanglement \cite{brunner,liang11,Vidick,Palazuelos}. Specifically, using weakly entangled states, we observe (i) nonlocal correlations which could provably not have been obtained from any maximally entangled state, and (ii) nonlocal correlations which could not have been obtained using a single PR box. Moreover, we observe the most nonlocal correlations ever reported, i.e. featuring the smallest local content \cite{EPR2}, and provide the strongest bounds on the outcome predictability in a general class of physical theories consistent with the no-signaling principle \cite{renato}. Finally, we observe nonlocal correlations which certify the use of complex qubit measurements. All our results are in remarkable agreement with quantum predictions.

\section{Concepts and notations}

First, we will introduce the concepts and notations for generalized Bell tests, and then present the experiments. Consider two separated observers, Alice and Bob, performing local measurements on a shared quantum state $\rho$. Alice's choice of measurement settings is denoted by $x$ and the measurement outcome by $a$. Similarly, Bob's choice of measurement is denoted by $y$ and its outcome by $b$. The experiment is thus characterized by the joint distribution
\ba \label{Q}
p(a,b|x,y) = \Tr(\rho M_{a|x} \otimes M_{b|y}),
\ea
where $M_{a|x}$ ($M_{b|y}$) represents the measurement operators of Alice (Bob); see Fig. 1(a). In his seminal work, Bell introduced a natural concept of locality, which assumes that the local measurement outcomes only depend on a pre-established strategy and the choice of local measurements \cite{bell}. Specifically, a distribution is said to be local if it admits a decomposition of the form
\ba \label{local}
p(a,b|x,y) = \int d \lambda\, q(\lambda) p(a|x,\lambda) p(b|y,\lambda) ,
\ea
where $\lambda$ denotes a shared local (hidden) variable, distributed according to the density $q(\lambda)$, and Alice's probability distribution---once $\lambda$ is given---is notably independent of Bob's input and output (and vice versa). For a given number of settings and outcomes the set of local distributions forms a polytope $\mathcal{L}$, the facets of which correspond to Bell inequalities \cite{review}. These inequalities can be written as 
\ba \label{BI} 
\mathcal{S} = \sum_{a,b,x,y} \beta_{a,b,x,y} p(a,b|x,y) \,\, \stackrel{\mathcal{L}}{\le} \,\, L,
\ea
where $\beta_{a,b,x,y}$ are integer coefficients, and $L$ denotes the local bound of the inequality---the maximum of the quantity $\mathcal{S}$ over distributions from $\mathcal{L}$, i.e., of the form \ref{local}. 

By performing judiciously chosen local measurements on an entangled quantum state, one can obtain distributions \ref{Q} which violate one (or more) Bell inequalities, and hence do not admit a decomposition of the form \ref{local}. Therefore, the set of quantum correlations $\mathcal{Q}$, i.e., those admitting a decomposition of the form \ref{Q}, is strictly larger than the local set $\mathcal{L}$. 
Characterizing the quantum set $\mathcal{Q}$, or equivalently the limits of quantum nonlocality, turns out to be a hard problem \cite{hierarchy1,hierarchy2}. In their seminal work, Popescu and Rohrlich \cite{PR} asked whether the principle of no-signaling (or relativistic causality) could be used to derive the limits of $\mathcal{Q}$ and surprisingly found this not to be the case. Specifically, they proved the existence of no-signaling correlations which are not achievable in quantum theory, the so-called ``PR box" correlations. Therefore, the set of no-signaling correlations, denoted by $\mathcal{NS}$, is strictly larger than $\mathcal{Q}$, and we get the relation $\mathcal{L} \subset \mathcal{Q} \subset \mathcal{NS}$ (see Fig. 1(b)). The study and characterization of the boundary between $\mathcal{Q}$ and $\mathcal{NS}$ is today a hot topic of research \cite{popescu14}. A central question is whether the limits of quantum nonlocality could be recovered from a simple physical principle (i.e., is it possible to derive quantum mechanics from just causality and another axiom).

\section{Experimental setup}

Here, we experimentally explore the limits of quantum nonlocality using a high quality source of entangled photons \cite{christensen}. Our entanglement source consists of a 355-nm pulsed laser focused onto two orthogonal nonlinear BiBO crystals to produce polarization-entangled photon pairs at 710 nm, via spontaneous parametric down-conversion: the first (second) crystal has an amplitude to create horizontal (vertical) polarized photon pairs, which interfere to produce the entangled state \cite{kwiatwaks}  (see Fig.~\ref{fig:expdia}).  Using wave plates to control the polarization of the pump beam, we create polarization entangled states with arbitrary degree of entanglement
\ba \ket{\psi_{\theta}} = \cos{\theta} \ket{H,H} + \sin{\theta} \ket{V,V}. \ea

In addition to the ability to precisely tune the entangled state of the source, which is crucial for many of the Bell tests we perform, we also achieve extremely high state quality.  To do so, we pre-compensate the temporal decoherence from group-velocity dispersion in the down-conversion crystals with a BBO crystal \cite{Rangarajan2009}, resulting in an interference visibility of $0.997\pm 0.0005$ in all bases.  The high state quality (along with the capability of creating a state with nearly any degree of entanglement) allows us to make measurements very close to the quantum mechanical bound in a large array of different Bell tests.

\begin{figure}[h]
\includegraphics[width=\columnwidth]{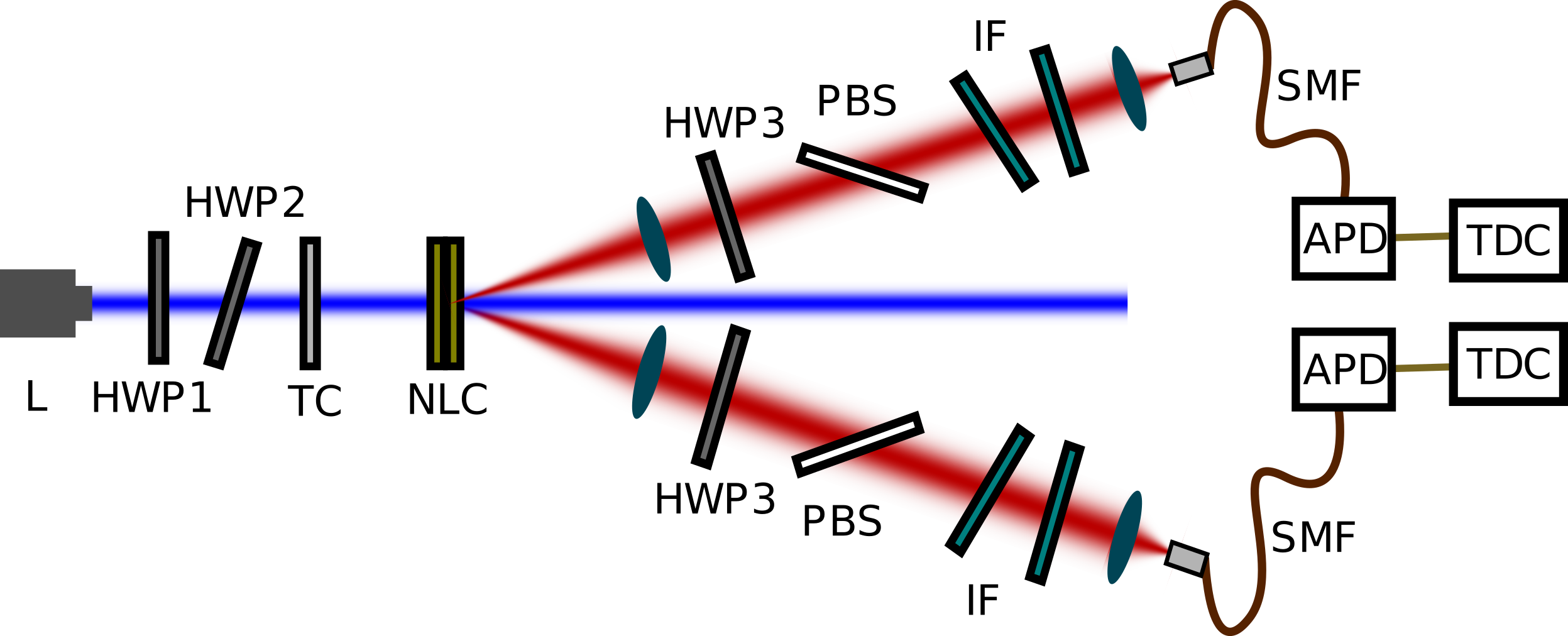} 
\caption{A diagram of the entanglement source. The high-power laser (L) is prepared in a specific polarization state (depending on the Bell test) by two half-wave plates (HWP1 and HWP2).  We pre-compensate for the temporal decoherence (arising from the group velocity dispersion in the downconversion crystals) by passing the laser through a crystal (TC) designed to have the opposite group velocity dispersion.  Passing the pump through a pair of orthogonal nonlinear crystals (NLC) produces the entangled photons. The measurements are performed using a motorized half-wave plate (HWP3) and a polarizing beam splitter (PBS). We then spectrally filter (IF) the photons to limit the collected bandwidth to 20 nm, as well as spatially filter the photons using a single-mode fiber (SMF) to remove any spatial decoherence. Finally, the photons are detected using avalanche photodiodes (APD), the events of which are recorded on a time-to-digital converter (TDC) and saved on a computer for analysis.}  \label{fig:expdia}
\end{figure}

For the Bell tests, the local polarization measurements are implemented using a fixed Brewster-angle polarizing beam splitter, preceded by an adjustable half-wave plate, and followed by single-photon detectors to detect the transmitted photons. This allows for the implementation of arbitrary projective measurements of the polarization, represented by operators $A = \vec{a} \cdot \vec{\sigma}$ and $B = \vec{b} \cdot \vec{\sigma}$, where $\vec{a} $ and $\vec{b} $ are the Bloch vectors and $\vec{\sigma} = (\sigma_x,\sigma_y,\sigma_z)$ denotes the vector of Pauli matrices. Measurement outcomes are denoted by $a=\pm1$ and $b=\pm1$, where in our experiments the $-1$ outcome is measured by projecting onto the orthogonal polarization.  To remove any potential systematic loopholes (e.g., seemingly better results due to laser power fluctuations), we measure each Bell inequality multiple times, where the measurements settings are applied in a different randomized order each time. Finally, to ensure the validity of the results, we do not perform any post-processing of the data (e.g., accidental subtraction).

\section{Experiments and results}


We start our investigation by considering the simplest Bell scenario, featuring two binary measurements each for Alice and Bob. The set of local correlations, i.e., of the form \ref{local}, is fully captured by the CHSH inequality \cite{chsh}: 
\ba 
\SCHSH = E_{11}+E_{12}+E_{21}-E_{22} \,\, \stackrel{\mathcal{L}}{\leq}\,\,  2,
\ea
where $E_{xy} \equiv p(a=b|x,y)-p(a \neq b|x,y) $ denotes the correlation function. Quantum correlations can violate the above inequality up to $\SCHSH= 2\sqrt{2}$, the so-called Tsirelson bound \cite{tsirelson}. More generally, quantum correlations must also satisfy the following family of inequalities 
\ba  \label{circle} 
	\SCHSH \cos{\theta} +   \SCHSH' \sin{\theta} \,\, \stackrel{\mathcal{Q}}{\leq}\,\,  2\sqrt{2},
\ea
parametrized by $\theta \in [0,2\pi]$, and where $\SCHSH'=-E_{11}+E_{12}+E_{21}+E_{22}$ is a different representation (or symmetry) of the CHSH expression.  
Notably, the above quantum Bell inequalities are tight, in the sense that quantum correlations can achieve $2\sqrt{2}$ for any $\theta \in [0,2\pi]$. Specifically, inequality \eqref{circle} can be saturated by performing appropriate local measurements on a maximally entangled state $\ket{\psi_{\pi/4}}$ (see App.~\ref{app:QB}). Therefore, the set of quantum correlations $\mathcal{Q}$ forms a circle (of radius $2\sqrt{2}$) in the plane defined by $\SCHSH$ and $\SCHSH'$. Fig. \ref{fig:CHSHcombined} presents the experimental results which confirm these theoretical predictions with high accuracy.  To make the measurements, we kept the entangled state fixed and varied the settings for 180 different values of $\theta$.  The average radius of our measurements was $2.817$, with a standard deviation of $0.006$, close to the quantum boundary of $2\sqrt{2}$ for the projection onto the $\SCHSH'$ and $\SCHSH$ axes.

\begin{figure*}[t]
\includegraphics[width=1.7\columnwidth]{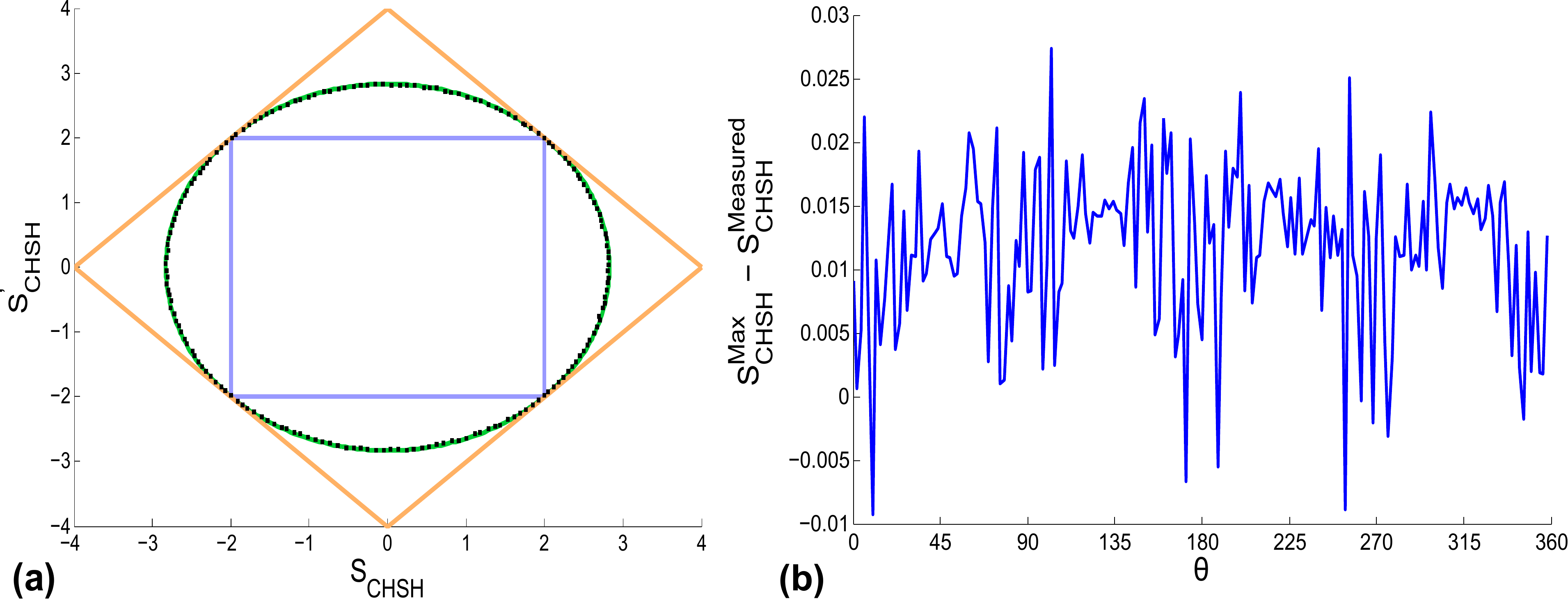} 
\caption{Testing the boundary of quantum correlations in the simplest Bell test scenario. (a) Plot of the experimental measurements of the curve Eq.~\eqref{circle}.  Here, local correlations ($\mathcal{L}$) form the inner blue square, no-signaling correlations form the outer orange square (the vertices represent the PR-box and its symmetries), and quantum correlations form the green circle; the black dots are the 180 measured data points, all of whose error bars lie within the thickness of the dot. (b) A comparison of the analyzed data with the quantum mechanical maximum ($2 \sqrt{2}$).  Here, $\theta$ is defined in Eq.~\eqref{circle}, and corresponds to rotating around the circle in Fig.~\ref{fig:CHSHcombined}(a). The vertical axis is the distance from $2 \sqrt{2}$ of the root mean square of $\SCHSH'$ and $\SCHSH$ (i.e., the radius of the data point at a given $\theta$). Plotted values greater than zero correspond to measured values less than $2 \sqrt{2}$.}\label{fig:CHSHcombined}
\end{figure*}

It turns out, however, that the complete boundary of $\mathcal{Q}$ cannot be fully recovered by considering only maximally entangled states. That is, there exist sections of the no-signaling polytope where the quantum boundary can only be reached using partially entangled states \cite{liang11} . Specifically, consider the projection plane defined by the parameters $\SCHSH$ and $-E^A_1-E^B_1$, where  $E^A_1= \sum_{a=\pm1} a \,\, p(a|x=1)$ denotes Alice's marginal (similarly for Bob's marginal $E^B_1$). In order to find the quantum boundary in this plane, we consider the family of Bell inequalities 
\ba\label{Ineq:TiltedCH}
	\mathcal{S}_{\tau} = \SCHSH  +  2(1-\tau)[E^A_1+E^B_1] \,\, \stackrel{\mathcal{L}}{\leq}\,\,  2(2\tau-1) ,
\ea
with $1\leq \tau \leq 3/2$. For $\tau=1$, we recover CHSH, while for $1< \tau < 3/2$ the inequality has the peculiar feature that the maximal quantum violation can only be obtained using partially entangled states \cite{liang11}. Moreover, for $1/\sqrt{2}+1/2 \leq \tau \leq \frac{3}{2}$, the inequality can \textit{never} be violated using a maximally entangled state of any Hilbert space dimension (see App.~\ref{App:MES}). This illustrates the fact that weak entanglement can give rise to nonlocal correlations which cannot be reproduced using strong entanglement. We achieved violations of the above inequalities (for several values of the parameter $\tau$) extremely close to the theoretically predicted maximum, by adjusting the degree of entanglement and using the corresponding settings; see Fig.~\ref{fig:Tiltedcombined}. For instance, tuning our source to produce weakly entangled states, we obtain clear violation of the inequality $\mathcal{S}_{\tau = 1.300} \stackrel{\mathcal{L}}{\le} 3.2$, where we measure $\mathcal{S}_{\tau} = 3.258 \pm 0.002$, which is impossible using maximally entangled states. Our results thus clearly illustrate the phenomenon of `more nonlocality with less entanglement'~\cite{Vidick,Palazuelos,liang11}.

\begin{figure*}[t]
\includegraphics[width=1.7\columnwidth]{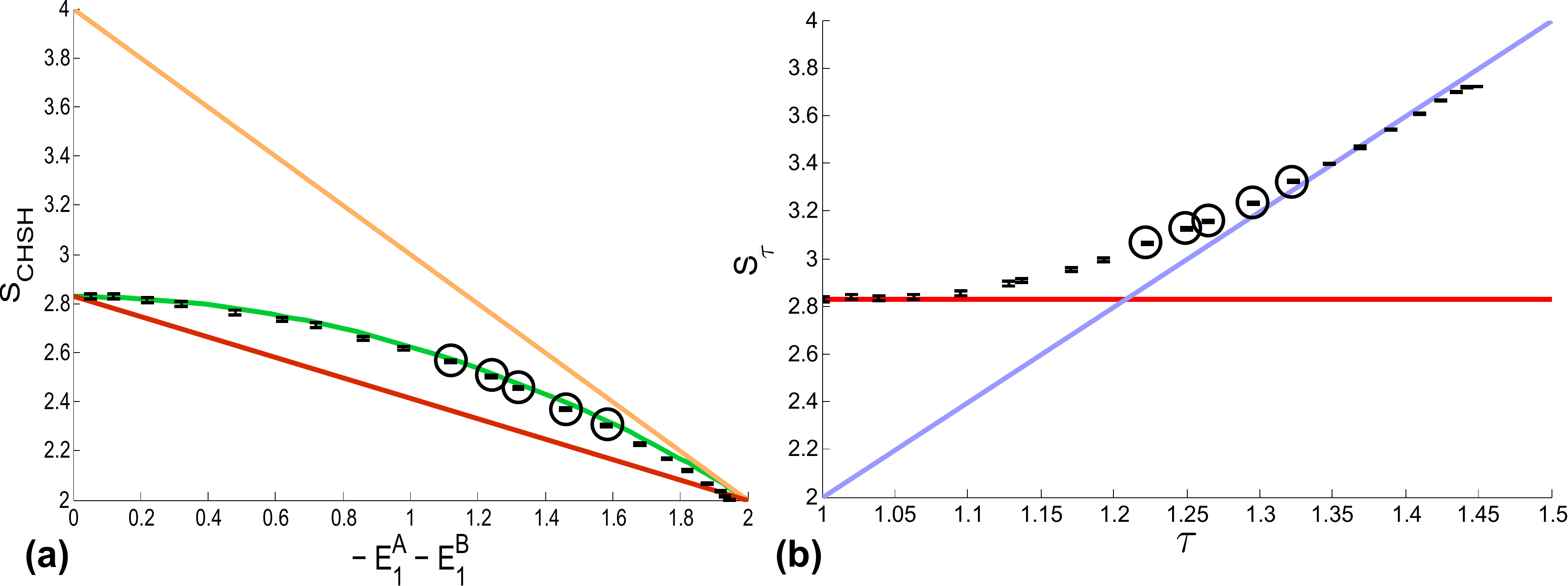}
\caption{Testing the boundary of quantum correlations with the tilted Bell inequality. (a) A plot of the measured values for a projection where the quantum boundary can only be attained using partially entangled states. The orange line is the boundary for no-signaling correlations (the PR box sitting at the top), and the red line is the boundary of the set of correlations achievable by a maximally entangled state (see App.~\ref{App:MES}), whereas the horizontal axis at $\SCHSH=2$ coincides with the boundary of the local set $\mathcal{L}$. The green curve represents the quantum boundary, with the black points corresponding to measured data points. For large values of $- E^A_1 - E^B_1$(corresponding to less entangled states), the system becomes increasingly sensitive to system noise (i.e., slight state-creation and measurement imperfections), resulting in the measured values deviating slightly from the quantum curve.  (b) A plot of the measured values for the tilted Bell inequalities. The red line is the bound of maximally entangled sates, the blue line is the local bound, and the black points are the analyzed data.  The red and blue lines cross at $1/\sqrt{2} + 1/2$, where maximally entangled states can no longer violate a tilted Bell inequality. Here, for the measured points up to $\tau = 1.323$, we see a value of $\mathcal{S}_{\tau}$ at least three standard deviations above the local bound; notably, we have violations for $\tau = 1.223, 1.250, 1.265, 1.296,$ and $1.323$ (circled data points in both plots), as well as $\tau = 1.300$ (see text and App.~\ref{App:Data}), none of which are possible for maximally entangled states in any dimension, implying that with less entanglement, we have more nonlocality.}\label{fig:Tiltedcombined}
\end{figure*}

In the remainder of the paper, we go beyond the CHSH scenario and consider Bell inequalities featuring $n>2$ binary-outcome measurements per observer. This will allow us to investigate other aspects of the phenomenon of quantum nonlocality. We start by considering the family of chained Bell inequalities \cite{pearle, BC}
\begin{align}\label{chainedineq}
	I_n=\sum_{a,b=\pm1}\bigg[&p(a=b|n,1) + p(a\neq b|n,n)  \nonumber \\ 
	&+\sum_{x=1}^{n-1}\sum_{y=x}^{x+1} p(a\neq b|x,y)\bigg] \stackrel{\mathcal{L}}{\ge} 1 .
\end{align}
Using a maximally entangled state $\ket{\psi_{\pi/4}}$, quantum theory allows one to obtain values up to $I_n=n(1-\cos\frac{\pi}{2n})$. Note that, as $n$ increases, the quantum violation approaches the bound imposed by the no-signaling principle, namely $I_n=0$  (here given by the algebraic minimum of $I_n$). Using our setup we obtain violations of the chained inequality up to $n=45$.  Because $I_n$ becomes increasingly sensitive to any noise in the system as $n$ increases, we found the strongest violation at $n=18$, with a value of $I_{18} = 0.126\pm 0.001$, see Fig.~\ref{fig:chained}.  For comparison, the previous best measurement of $I_n$ was $I_7 = 0.324 \pm 0.0027$ \cite{stuart}.

\begin{figure}[t]
\includegraphics[width=\columnwidth]{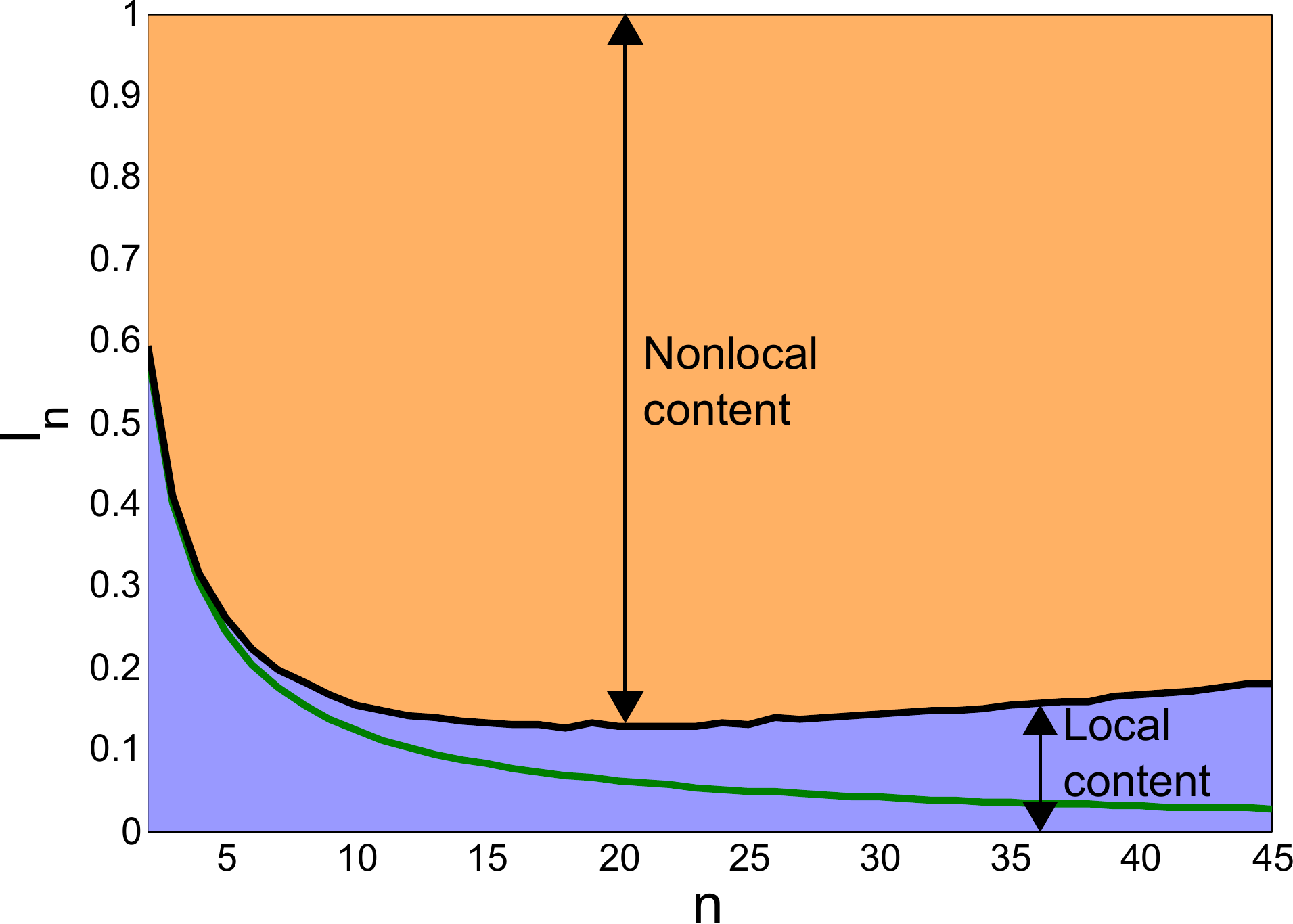} 
\caption{\label{fig:chained} A plot of the measured chained Bell inequality values for $n=2$ to $n=45$.  Here, the local limit is $I_n = 1$ and the no-signaling limit is $I_n = 0$.  The quantum boundary in this case is the green line; our measured Bell inequality values are connected by the black line, with the error bars lying within the thickness of the line.  The local content for a given $n$ is represented by distance from $0$ to the measured $I_n$ value (black line), which is colored blue, and the nonlocal content is the distance from the measured value to $1$, colored orange.  As the value of $I_n$ approaches 0, the correlations present in the system match those of a PR box---if $I_n = 0$ were measured, the system would require the use of a PR box for every measurement.  Our measured points deviate from the quantum boundary due to the 0.3\% noise from imperfect state preparation, which becomes more noticeable with larger number of measurements (e.g., $I_{45}$ requires 360 specific measurements along the Bloch sphere).}
\end{figure}

These violations have interesting consequences. First, they allow us to put strong lower bounds on the nonlocal content of the observed statistics ${\bf p}_{obs} = \{p_{obs}(ab|xy)\}$. Following the approach of Ref. \cite{EPR2}, we can write the decomposition
\ba {\bf p}_{obs} = (1-q) {\bf p}_{L} + q\, {\bf p}_{NS}, \ea
where ${\bf p}_{L}$ is a local distribution (inside $\mathcal{L}$) and ${\bf p}_{NS}$ is a no-signaling distribution (achieved, e.g., via PR boxes), and then minimize $q \,\in[0,1]$ over any such decomposition. The minimal value $q_{min}$ is then the nonlocal content of ${\bf p}_{obs} $, and can be viewed as a measure of nonlocality.  That is, we can think of $q_{min}$ as being the likelihood that some nonlocal resource (e.g., a PR box) would need to be used in order to replicate the results. For an observed violation of the chained inequality, we can place a lower bound on the nonlocal content: $q_{min} \geq 1-I_n$ \cite{nonlocalcontent}. Notably, for the case $n=18$, we obtain $q_{min} = 0.874\pm 0.001$ which represents the most nonlocal correlations ever produced experimentally. For comparison, the previous best bound was $q_{min}=0.782\pm 0.014$ \cite{previous_nonlocal1,previous_nonlocal2} (and if one maximally violates $\SCHSH$ , then $q_{min}=0.41$).

Moreover, following the work of \cite{renato}, we can place bounds on the outcome predictability in a general class of physical theories consistent with the no-signaling principle. While quantum theory predicts that the measurement results are fully random (e.g., one cannot predict locally which output port of the polarizing beam splitter each photon will be detected), there could be a super-quantum theory that could predict better than quantum theory (that is with a probability of success strictly greater than $1/2$) in which port each photon will be detected. This predictive power, represented by the probability $\delta$ of correctly guessing the output port, can be upper bounded from the observed violation of the chained Bell inequality. In our experiments, the best bound is obtained for the case $n=18$, for which we obtain $\delta = 0.5702 \pm 0.0005$ (that is, given any possible extension of quantum theory satisfying the free-choice assumption~\cite{stuart}, the measurement result could be guessed with a probability at most 57\% ), which is the strongest experimental bound (closest to 50\% ) to date; the previous bound was $\delta = 0.6644 \pm 0.0014$ \cite{stuart}.


\begin{table}[t]
\caption{\label{tab:MBell}A table of the measured values from two different Bell inequalities, $M_{3322}$ and $M_{4322}$, as defined in Eq.~\eqref{PRBoxIneq}. 
For these inequalities, correlations from $\mathcal{L}$ and those augmented with the use of a single PR-box (represented as $\mathcal{L}+1PR$) give rise to the same bound. Any measured values above the corresponding bound imply that the data is not only incompatible with Bell-locality, but also with a single use of a PR box.  Instead, two PR boxes must be used to replicate the data.  The approximate quantum mechanical maximums (obtained using the tools of~\cite{hierarchy1,hierarchy2} and~\cite{yc-tools1,yc-tools2}) and the quantum mechanical maximums for two qubits are given as a reference.}
\begin{ruledtabular}
\begin{tabular}{lcc}
Bell inequality & Measured value & Quantum (2-qubit) maximum \\
\hline
$M_{3322} \stackrel{ \mathcal{L}+1PR}{\leq} 6$ & $6.016 \pm 0.0003$ &  6.130 (6.024)\\
$M_{4322}\stackrel{ \mathcal{L}+1PR}{\leq} 7$ & $7.004 \pm 0.0004$ &  7.127 (7.041)\\
\end{tabular}
\end{ruledtabular}
\end{table}

The above results on the chained Bell inequality show that in order to reproduce the measured correlations, nonlocal resources (such as the PR box) must be used in more than $87\%$ of the experimental rounds. While the chained Bell inequality provides an interesting metric of nonlocal content, there are, however, even more nonlocal correlations achievable using two-qubit entangled states, which can provably not be reproduced using a single PR box \cite{brunner}. Interestingly, such correlations can arise only from partially entangled states, since maximally entangled states can always be perfectly simulated using a single PR box \cite{cerf05}. The accuracy of our experimental setup allows for the study of Bell inequalities which require the use of more than a single PR box. Specifically, consider the inequalities from Ref.~\cite{brunner} (for $n=3$ and $n=4$): 
\begin{subequations}\label{PRBoxIneq}
\begin{align} 
\mathcal{M}_{3322} = &E_{11}+E_{12}+E_{13}+E_{21}+E_{22}-E_{23} \\ \nonumber  
						+&E_{31}-E_{32}-E^A_1-E^A_2 -E^B_1 + E^B_2 \\ \nonumber
						&\!\!\!\!\!\!\!\!\!\!\!\!\stackrel{\mathcal{L}+1PR}{\leq} 6,\\
\mathcal{M}_{4322} = &E_{11}+E_{12}+E_{13}+E_{21}-E_{23}+E_{24} \\ \nonumber  
						+ &E_{31}-E_{32}-E_{34}-E^A_1-E^A_2 -E^A_3 - E^B_1 \\ \nonumber
						&\!\!\!\!\!\!\!\!\!\!\!\!\stackrel{\mathcal{L}+1PR}{\leq} 7,
\end{align}
\end{subequations}
which cannot be violated by any local correlations supplemented by a single maximally nonlocal PR box ($\mathcal{L}+1PR$), which is viewed as a unit of nonlocality. Nevertheless, by performing well-chosen measurements on a very weakly entangled state ($\ket{\psi_{\approx 3\pi/7}}$), we observed violations of the above inequalities (see Table~\ref{tab:MBell}). Note that since the observed statistics (leading to $\mathcal{M}_{3322} > 6$ and $\mathcal{M}_{4322} > 7$) could not have been obtained using a single PR box, they also cannot be obtained using a maximally entangled state $\ket{\psi_{\pi/4}}$, and required the use of a weakly entangled state (or two PR boxes). Hence, we provide a second experimental verification of the phenomenon of more nonlocality with less entanglement.

Finally, we consider a Bell inequality which can certify the use of complex qubits (versus real qubits) \cite{gisin07}. Specifically, the Bell inequality is given by 
\ba 
& \mathcal{S}_{E} = E_{11}+E_{12}+E_{13}+E_{21}-E_{22}-E_{23} \\ \nonumber  
						& \quad \quad \quad \quad -E_{31}+E_{32}-E_{33}-E_{41}-E_{42}+E_{43} 
						\stackrel{\mathcal{L}}{\leq} 6.
\ea
The optimal quantum violation is $\mathcal{S}_{E}=4 \sqrt{3} \simeq 6.928$, which can be obtained by using a maximally entangled two-qubit state $\ket{\psi_{\pi/4}}$, and a set of highly symmetric qubit measurements. The measurements of Bob are given by three orthogonal vectors on the sphere, and Alice's measurements are given by the four vectors of the tetrahedron: $\vec{a}_1 = \frac{1}{\sqrt{3}}(1,1,1)$, $\vec{a}_2 = \frac{1}{\sqrt{3}}(1,-1,-1)$, $\vec{a}_3 = \frac{1}{\sqrt{3}}(-1,1,-1)$, $\vec{a}_4 = \frac{1}{\sqrt{3}}(-1,-1,1)$, and $\vec{b}_1 = (1,0,0)$, $\vec{b}_2 = (0,1,0)$, $\vec{b}_3 = (0,0,1)$. Implementing this strategy experimentally, we observe a violation of $\mathcal{S}_{E} = 6.890 \pm 0.002$, close to the theoretical value. Interestingly, such a strong violation could not have been obtained using a real qubit strategy. Indeed, the use of measurement settings restricted to an equator of the Bloch sphere, i.e., real qubit measurements, can only provide violations up to $\mathcal{S}_{E} = 2 +2\sqrt{5} \simeq 6.472$ \cite{gisin07}. Note, however, that any strategy involving a single complex qubit measurement can be mapped to an equivalent strategy involving two real qubits \cite{mosca,pal}. Thus, the observed violation certifies the use of complex qubit measurements, i.e., spanning the Bloch sphere, or the use of a higher-dimensional real Hilbert space \cite{FN_dim}.

\section{Conclusion}

To summarize, we have reported the observation of various facets of the rich phenomenon of quantum nonlocality. The results of our systematic experimental investigation of quantum nonlocal correlations are in extremely good agreement with quantum predictions; nevertheless, we believe that pursuing such tests is of significant value, as Bell inequalities are not only fundamental to quantum theory, but also can be used to discuss physics outside of the framework of quantum theory. By doing so, one can continue to place bounds on the features of theories beyond quantum mechanics, as we have here. Such continued experiments investigating the bounds of quantum theory are important, as any valid deviation with quantum predictions, e.g., by observing stronger correlations than predicted by quantum theory, would provide evidence of new physics beyond the quantum model. Furthermore, nonlocality has important applications towards quantum information protocols, though the optimal way to quantify the nonlocality present in a system is still an open question (see, e.g.,~\cite{Bernhard2014}). Here, we experimentally verified, for the first time, that for certain correlations from non-maximally entangled states, two PR boxes (i.e., two units of the nonlocal resource) are required to recreate the correlations from these weakly entangled states. A natural question then is if these systems could be used advantageously for certain quantum information tasks.

\begin{acknowledgments}
This research was supported by the NSF grant No. PHY 12-05870, the Ministry of Education, Taiwan, R.O.C., through ``The Aim for the Top University Project" granted to the National Cheng Kung University (NCKU), the Swiss NCCR-QSIT, the Swiss National Science Foundation (grant PP00P2\_138917 and Starting Grant DIAQ), SEFRI (COST action MP1006).
\end{acknowledgments}

Corresponding author: Bradley Christensen (bgchris2@illinois.edu).

\appendix

\section{Characterization of the quantum boundary}
\label{app:QB}

Here, we discuss in detail the characterization of the boundary of quantum correlations in a 2-dimensional projection of the no-signaling polytope in the case of binary inputs and outputs (i.e., the CHSH scenario). 
Specifically, let us consider the 2-dimensional plane (Fig.~\ref{fig:CHSHcombined}) defined by the expectation values of:
\ba\label{Eq:Axes}
	\SCHSH &=& E_{11}+E_{12}+E_{21}-E_{22}, \\
	\SCHSH' &=& -E_{11}+E_{12}+E_{21}+E_{22}.
\ea
Note that the correlation functions $E_{xy}$ can equivalently be seen as the average value of the product of $\pm1$-outcome local (projective) measurements, i.e., 
\begin{equation}\label{Eq:Exy}
	E_{xy}=\sum_{a,b=\pm1} a\,b\, p(a,b|x,y).
\end{equation}
They can thus be evaluated in quantum theory as $E_{xy}={\rm tr}(\rho\,A_xB_y)$ where
$A_x$, $B_y$ are dichotomic observables satisfying 
\begin{equation}\label{Eq:Constraints}
	A_x^2=B_y^2=\one, \quad [A_x,B_y]=0\quad\forall\quad x,y.
\end{equation}
The boundary of the set of legitimate quantum distributions in this 2-dimensional plane is given by the circle (see also~\cite{Allcock2009,Branciard2011}):
\begin{equation}\label{Eq:Circle}
	\SCHSH^2 + \SCHSH'^2 \leq 8,
\end{equation}
or equivalently
\begin{equation}\label{Eq:Linearized}
	\SCHSH \cos\theta + \SCHSH' \sin\theta \leq 2\sqrt{2} \quad \forall\quad \theta\in\left[0,2\pi\right].
\end{equation}
To see that this is the case, let us note that for any dichotomic observables $A_1$, $A_2$, $B_1$, $B_2$ satisfying Eq.~\eqref{Eq:Constraints},  and any $\theta\in\left[0,2\pi\right]$, the following identity holds true:
\ba \label{Eq:SOS}
	& & 2\sqrt{2}\,\one - \mathcal{B} \\ \nonumber 
	&= &\frac{1}{\sqrt{2}}\left[\sin\left(\frac{\pi}{4}+\theta\right)A_2+\cos\left(\frac{\pi}{4}+\theta\right)A_1-B_1\right]^2\\ \nonumber
	&+&\frac{1}{\sqrt{2}}\left[\cos\left(\frac{\pi}{4}+\theta\right)A_2-\sin\left(\frac{\pi}{4}+\theta\right)A_1+B_2\right]^2
\ea
where $\mathcal{B}$ is the Bell operator~\cite{Braunstein} associated with the Bell expression given in the left-hand-side of Eq.~\eqref{Eq:Linearized}, i.e.,
\ba
	\mathcal{B} &=& \sum_{x,y=1}^2 \left\{\left[\cos\theta(-1)^{(x-1)(y-1)} +\sin\theta (-1)^{xy} \right]A_xB_y\right\}, 
	\nonumber \\
	 &=&\sqrt{2}\Big[\sin\left(\frac{\pi}{4}+\theta\right)(A_1B_2+A_2B_1)\nonumber\\
	 & & \qquad\,\,+\cos\left(\frac{\pi}{4}+\theta\right)(A_1B_1-A_2B_2)\Big].
\ea
Since the right-hand-side of Eq.~\eqref{Eq:SOS} is a sum of non-negative operators, it then follows that for any quantum state $\rho$, and hence for any quantum correlation, we must have
\begin{equation}
	2\sqrt{2}-\Tr(\rho \, \mathcal{B})=2\sqrt{2}-(\SCHSH \cos\theta + \SCHSH' \sin\theta) \geq 0.
\end{equation}
The above bound on the set of quantum correlations is indeed achievable. Explicitly, for each $\theta$, by measuring the following observables:
\begin{gather}
	A_1=\sigma_x, \quad A_2=\sigma_z,\\
	B_1=-\sin\chi\,\sigma_z-\cos\chi\,\sigma_x,\quad 
	B_2=\cos \chi\,\sigma_z-\sin\chi\,\sigma_x,\nonumber
\end{gather}
with $\chi=\theta-\frac{3\pi}{4}$ on the maximally entangled two-qubit state
\begin{equation}
	\ket{\psi_{\pi/4}}=\frac{1}{\sqrt{2}}(\ket{H,H}+\ket{V,V})
\end{equation}
[with $\ket{H}$ ($\ket{V}$) being the $+1 (-1)$ eigenstate of $\sigma_z$], 
we arrive at quantum correlations that saturate the inequality given in Eq.~\eqref{Eq:Linearized}. By 
varying $\theta$ over the entire interval $[0,2\pi]$, it can be verified that one indeed generates the entire (circular) boundary of the quantum set in this 2-dimensional plane.

\section{Upper-bounding quantum violation by maximally entangled states }
\label{App:MES}

Let us denote by $\tau_{\mbox{\tiny Cr}}=\frac{1}{2}+\frac{1}{\sqrt{2}}\approx 1.2071$ and by $\ket{\Phi^+_{d}}$ the maximally entangled state of local Hilbert space dimension $d$, i.e., 
\begin{equation}\label{Eq:MES}
	\ket{\Phi^+_{d}}=\frac{1}{\sqrt{d}}\sum_{i=0}^{d-1}\ket{i}\ket{i},
\end{equation}
where $d\ge 2$. Here, we provide further details showing the following observation.
\begin{observation}\label{Thm:TiltedCHMES}
For $ \tau\in[\tau_{\mbox{\tiny Cr}}, \frac{3}{2}]$, the family of Bell inequalities given by Eq.~\eqref{Ineq:TiltedCH}  cannot be violated by any finite-dimensional maximally entangled state $\ket{\Phi^+_{d}}$.
\end{observation}

\begin{proof}
Let us first note that for arbitrary $\tau\in[\tau_{\mbox{\tiny Cr}},\frac{3}{2}]$, the Bell inequality $\Stau$ can be written as a convex combination of that for $\tau=\tau_{\mbox{\tiny Cr}}$ and that for $\tau=\frac{3}{2}$. Moreover, for any given quantum state $\rho$ (and given experimental scenario), it is easy to show that the set of Bell inequalities {\em satisfied} by $\rho$ is a convex set. Since $\Stau$ for $\tau=\frac{3}{2}$ cannot be violated by any legitimate probability distribution~\cite{liang11},  it suffices to show that $\Stau$  for $\tau=\tau_{\mbox{\tiny Cr}}$ also cannot be violated by $\ket{\Phi^+_{d}}$ (for any finite $d$).

To show that {\em no} finite-dimensional $\ket{\Phi^+_{d}}$  can violate the  Bell inequality $\Stau$ for $\tau=\tau_{\mbox{\tiny Cr}}$, we  make use of the hierarchy of semidefinite programs (SDPs) considered in  Ref.~\cite{Lang:JPA:424029} for characterizing {\em exactly} the quantum correlations achievable by $\ket{\Phi^+_{d}}$.  Specifically, to obtain (an upper bound on) the maximal value of $\mathcal{S}_{\tau=\tau_{\mbox{\tiny Cr}}}$ attainable by finite-dimensional $\ket{\Phi^+_{d}}$, it suffices to consider a fixed level of the hierarchy, and solve an SDP over the positive semidefinite (moment) matrix variable $\Gamma$ such that  (1) certain entries of $\Gamma$ are required to be non-negative, (2) certain entries of $\Gamma$ are required to be identical, and (3) a particular entry of $\Gamma$ is required to be ``1" (for details, see page 16-17 of Ref.~\cite{Lang:JPA:424029}).

To this end, we consider $\Gamma$ defined by the 17 symbolic operators $\mathbbm{1}$, $A^+_x$, $B^+_y$, $A^+_xB^+_y$, $B^+_yA^+_x$, $B^+_yB^+_{y'}$, $A^+_xA^+_{x'}$, with $x, x', y, y'\in{1,2}$ and $x\neq x'$, $y \neq y'$.
Solving the corresponding SDP via the solver ``sedumi" (interfaced through YALMIP~\cite{yalmip}), we found that the quantum value of $\mathcal{S}_{\tau=\tau_{\mbox{\tiny Cr}}}$ attainable by any finite-dimensional $\ket{\Phi^+_{d}}$ is upper bounded by $2(2\tau_{\mbox{\tiny Cr}}-1)+\epsilon$ with $\epsilon\approx1.09\times10^{-8}$, which is vanishing within the numerical precision of the solver. In other words, after accounting for the numerical error present in the optimization problem, the output of the SDP provides a numerical certificate that no finite-dimensional maximally entangled state can violate the  Bell inequality $\Stau$ for $\tau=\tau_{\mbox{\tiny Cr}}$. This completes the proof of Observation~\ref{Thm:TiltedCHMES}.
\end{proof}

Note that as $\tau$ increases, the equality $\Stau=2(2\tau-1)$, cf.  Eq.~\eqref{Ineq:TiltedCH} corresponds to a plane tilting from the horizontal axis towards the vertical axis at $-E^A_1-E^B_1=2$, see Fig.~\ref{fig:Tiltedcombined}(a). For $\tau=\tau_{\mbox{\tiny Cr}}$, this plane corresponds precisely to the red solid straight line joining the points $(0,2\sqrt{2})$ and $(2,0)$. Observation~\ref{Thm:TiltedCHMES} thus translates to the fact that, in Fig.~\ref{fig:Tiltedcombined}(a), all correlations present in the region between the actual quantum boundary (green curve) and the red solid line are unattainable by finite-dimensional maximally entangled states --- a fact that can also be independently verified by solving an analogous SDP which maximizes the value of $\SCHSH$ under the equality constraint that marginal correlations $-E^A_1-E^B_1$ take on specific values in the interval $[0,2]$.

Finally, let us note that the same numerical technique could also be used to show that both the $M_{3322}$ and the $M_{4322}$ inequality hold true (to within a numerical precision of $10^{-8}$) for all correlations arising from maximally entangled state of any Hilbert space dimension.

\section{Detailed data and estimated states}
\label{App:Data}

In this section, we give the results of the analyzed data and quantum states used for each Bell test presented in this paper. First, for the data collected for the projection onto the $\SCHSH \cos{\theta}$ and $\SCHSH' \sin{\theta}$ axes, we used maximally entangled states, altering the measurement settings to rotate around the circle in Fig.~\ref{fig:CHSHcombined}.  Here, we collected data for 1s at each setting (where each point requires 16 total measurement combinations).

For the plot of the tilted Bell inequality in Fig.~\ref{fig:Tiltedcombined} [Eq.~\eqref{Ineq:TiltedCH}], we collected data for 15 s for each setting (again, 16 total measurement setting combinations).  We used states of varying degree of entanglement, which we cite by listing the $\theta$ value in the state $\cos{\theta}|H,H\rangle + \sin{\theta}|V,V\rangle$. The analyzed data is displayed in Table \ref{tab:tilteddata}.  As a note, the value in the text listed for $\mathcal{S}_{\tau = 1.3}$ had separately optimized settings (instead of automatically generated settings), as well as was measured for 100 s.

\begin{table}
\caption{\label{tab:tilteddata}Analyzed data and estimated parameters for the tilted Bell inequality [Eq.~\eqref{Ineq:TiltedCH}]. Here, the estimate of the uncertainty of $\Stau$ is given by $\Delta \Stau$.}
\begin{ruledtabular}
\begin{tabular}{cccccccc}
$\tau$ & $\SCHSH$ & $- E_1 - E_2$ & $\Stau$ & $\Delta \Stau$ & Local bound & $\theta$\\
\hline

    1.001  &  2.828  &  0.052  &  2.828  &  0.011  &  2.004  &  45.0\\
    1.020  &  2.827  &  0.120  &  2.837  &  0.010  &  2.080  &  46.5\\
    1.039  &  2.816  &  0.220  &  2.833  &  0.010  &  2.156  &  48.1\\
    1.063  &  2.800  &  0.320  &  2.840  &  0.010  &  2.252  &  49.8\\
    1.095  &  2.764  &  0.480  &  2.855  &  0.009  &  2.380  &  52.2\\
    1.128  &  2.736  &  0.620  &  2.895  &  0.009  &  2.512  &  54.9\\
    1.137  &  2.712  &  0.720  &  2.909  &  0.008  &  2.548  &  55.5\\
    1.171  &  2.660  &  0.860  &  2.954  &  0.008  &  2.684  &  58.5\\
    1.193  &  2.616  &  0.980  &  2.994  &  0.007  &  2.772  &  60.2\\
    1.223  &  2.564  &  1.120  &  3.064  &  0.007  &  2.892  &  62.7\\
    1.250  &  2.504  &  1.240  &  3.124  &  0.006  &  3.000  &  65.2\\
    1.265  &  2.456  &  1.320  &  3.156  &  0.006  &  3.060  &  66.4\\
    1.296  &  2.368  &  1.460  &  3.232  &  0.005  &  3.184  &  69.2\\
    1.323  &  2.304  &  1.580  &  3.325  &  0.005  &  3.292  &  71.7\\
    1.348  &  2.228  &  1.680  &  3.397  &  0.004  &  3.392  &  74.2\\
    1.369  &  2.168  &  1.760  &  3.467  &  0.003  &  3.476  &  76.3\\
    1.390  &  2.120  &  1.820  &  3.540  &  0.003  &  3.560  &  78.4\\
    1.410  &  2.064  &  1.880  &  3.606  &  0.002  &  3.640  &  80.5\\
    1.424  &  2.036  &  1.920  &  3.664  &  0.002  &  3.696  &  81.9\\
    1.435  &  2.016  &  1.932  &  3.697  &  0.002  &  3.740  &  82.9\\
    1.442  &  2.000  &  1.946  &  3.721  &  0.002  &  3.768  &  83.7\\
    1.449  &  1.964  &  1.957  &  3.722  &  0.002  &  3.796  &  84.6\\
\end{tabular}
\end{ruledtabular}
\end{table}

For the chained Bell inequality [Eq.~\eqref{chainedineq}], $I_n$ is the chained Bell parameter, with $\nu_n$ being the measurement bias.  The measurement bias is the deviation of Alice's (or Bob's) individual measurements from being completely random, that is, the difference in probability of measuring output -1 to measuring output 1 (calculated by $p(1|x)-p(-1|x)$). The bias given in Table~\ref{tab:chaineddata} (and the bias used in calculating $\delta_n$) is the maximum bias over all possible measurement settings.  Finally, $\delta_n$ is the bound on the predictive power, with $\Delta \delta_n$ being the uncertainty.  The uncertainty of $I_n$ is approximately twice as large as the $\Delta \delta_n$ (since $\delta_n \propto I_n/2$).  For these measurements, we used a maximally entangled state, and collected data for 5 s at each measurement setting, except from $n = 18$ to $n = 21$, where we collected for 20 s at each setting, as the first scan through all values of $n$ showed the lowest value in that region.  The analyzed data for the chained Bell inequality is shown in Table \ref{tab:chaineddata}.
\begin{table}
\caption{\label{tab:chaineddata}Analyzed data for the chained Bell inequality [Eq.~\eqref{chainedineq}].}
\begin{ruledtabular}
\begin{tabular}{ccccc}
$n$ & $I_n$ & $\nu_n$ & $\delta_n$ & $\Delta \delta_n$\\
\hline
    2  &  0.5931  &  0.0062  &  0.8028  &  0.0016\\
    3  &  0.4115  &  0.0058  &  0.7116  &  0.0014\\
    4  &  0.3148  &  0.0055  &  0.6629  &  0.0013\\
    5  &  0.2624  &  0.0068  &  0.6380  &  0.0012\\
    6  &  0.2230  &  0.0058  &  0.6173  &  0.0012\\
    7  &  0.1965  &  0.0065  &  0.6048  &  0.0011\\
    8  &  0.1812  &  0.0059  &  0.5964  &  0.0011\\
    9  &  0.1667  &  0.0073  &  0.5906  &  0.0011\\
   10  &  0.1539  &  0.0066  &  0.5836  &  0.0011\\
   11  &  0.1479  &  0.0069  &  0.5809  &  0.0011\\
   12  &  0.1419  &  0.0069  &  0.5778  &  0.0011\\
   13  &  0.1396  &  0.0065  &  0.5763  &  0.0011\\
   14  &  0.1357  &  0.0064  &  0.5742  &  0.0011\\
   15  &  0.1324  &  0.0077  &  0.5739  &  0.0010\\
   16  &  0.1312  &  0.0061  &  0.5718  &  0.0010\\
   17  &  0.1294  &  0.0064  &  0.5711  &  0.0010\\
   18  &  0.1262  &  0.0065  &  0.5702  &  0.0005\\
   19  &  0.1318  &  0.0070  &  0.5714  &  0.0005\\
   20  &  0.1290  &  0.0075  &  0.5722  &  0.0005\\
   21  &  0.1279  &  0.0074  &  0.5709  &  0.0005\\
   22  &  0.1291  &  0.0071  &  0.5717  &  0.0010\\
   23  &  0.1287  &  0.0065  &  0.5708  &  0.0010\\
   24  &  0.1325  &  0.0072  &  0.5734  &  0.0010\\
   25  &  0.1312  &  0.0074  &  0.5730  &  0.0010\\
   26  &  0.1380  &  0.0067  &  0.5757  &  0.0011\\
   27  &  0.1372  &  0.0070  &  0.5755  &  0.0010\\
   28  &  0.1389  &  0.0073  &  0.5768  &  0.0011\\
   29  &  0.1409  &  0.0073  &  0.5777  &  0.0011\\
   30  &  0.1429  &  0.0069  &  0.5783  &  0.0011\\
   31  &  0.1456  &  0.0075  &  0.5803  &  0.0011\\
   32  &  0.1474  &  0.0066  &  0.5803  &  0.0011\\
   33  &  0.1475  &  0.0070  &  0.5808  &  0.0011\\
   34  &  0.1506  &  0.0083  &  0.5836  &  0.0011\\
   35  &  0.1547  &  0.0073  &  0.5846  &  0.0011\\
   36  &  0.1573  &  0.0066  &  0.5853  &  0.0011\\
   37  &  0.1577  &  0.0081  &  0.5870  &  0.0011\\
   38  &  0.1594  &  0.0072  &  0.5869  &  0.0011\\
   39  &  0.1655  &  0.0072  &  0.5899  &  0.0011\\
   40  &  0.1665  &  0.0070  &  0.5903  &  0.0011\\
   41  &  0.1698  &  0.0073  &  0.5922  &  0.0011\\
   42  &  0.1716  &  0.0065  &  0.5923  &  0.0011\\
   43  &  0.1750  &  0.0067  &  0.5942  &  0.0011\\
   44  &  0.1810  &  0.0069  &  0.5974  &  0.0011\\
   45  &  0.1801  &  0.0079  &  0.5980  &  0.0011\\
\end{tabular}
\end{ruledtabular}
\end{table}

Finally, in Table~\ref{tab:mdata} we list the measurement settings and states for the $M_{3322}$ and $M_{4322}$ Bell inequalities.  The settings are given as the angle in the projection onto the state $\cos{a_i}|H\rangle + \sin{a_i}|V\rangle$ (and similarly for $b_j$). Here, data was collected data for 1200 s at each measurement setting.
\begin{table}
\caption{\label{tab:mdata}Details for the $M_{3322}$ and $M_{4322}$ Bell inequalities [Eqs.~\eqref{PRBoxIneq}].}
\begin{ruledtabular}
\begin{tabular}{ccccccccc}
Inequality & $\theta$ & $a_1$ & $a_2$ & $a_3$ & $a_4$ & $b_1$ & $b_2$ & $b_3$\\
\hline

    $M_{3322}$  & 77.2 & -1.2 & 27.2 & -35.2 & N/A & -0.7 & 9.2 & -20.3\\
    $M_{4322}$  & 76.6 & 0 & 61 & 45 & 119 & 15.6 & 164.3 & 0\\
\end{tabular}
\end{ruledtabular}
\end{table}


\begin{thebibliography}{99}
\bibitem{FN_dim} It is known that a $2\otimes 2$ system cannot produce a violation, whereas a $4\otimes 4$ system can.  However, it is not clear whether a smaller dimensional system (e.g., $2\otimes 3$ or $3\otimes 4$) is sufficient to produce a violation.
\end{thebibliography}
\end{document}